\documentclass{llncs}                          
\usepackage{times}                               
\usepackage{latexsym}                              

\newcommand{\lambdax}{\lambda \hspace{-0.30mm} x}

\newcommand{\lambday}{\lambda \hspace{-0.55mm} y}

\newcommand{\Sort}{{\cal S}}

\newcommand{\Type}{{\cal T}_{\cal S}}
\newcommand{\gttype}{>_{\Type}}
\newcommand{\getype}{\geq_{\Type}}
\newcommand{\letype}{\leq_{\Type}}
\newcommand{\eqtype}{=_{\Type}}
\newcommand{\neqtype}{\neq_{\Type}}
\newcommand{\gtS}{\gttype}
\newcommand{\geS}{\getype}
\newcommand{\leS}{\letype}
\newcommand{\eqS}{\eqtype}

\newcommand{\Term}{\TT}

\newcommand{\gtF}{>_{\cal F}}

\newcommand{\geF}{\geq_{\cal F}}

\newcommand{\comment}[1]{}

\newcommand{\cqfd}{\hfill $\Box$}

\newcommand{\vect}[1]{\overline{#1}}

\newcommand{\ra}{\rightarrow}

\newcommand{\Nat}{\mbox{$\mbox{I}\!\mbox{N}$}}

\newcommand{\TFX}{\TT(\F,\X)}

\newcommand{\TT}{{\cal T}}
\newcommand{\F}{{\cal F}}

\newcommand{\X}{{\cal X}}

\newcommand{\Var}[1]{{\cal V}ar({#1})}

\newcommand{\lrps}[2]{\mathop{\longrightarrow}^{#1}_{#2}}

\newcommand{\gtpo}{\succ}
\newcommand{\gepo}{\succeq}

\newcommand{\gtsubt}{\rhd}
\newcommand{\gesubt}{\unrhd}

\newenvironment{ruleset}
               {$$ \begin{array}{llcr}}{\end{array} $$ }

\newcommand{\Update}[1]{
{\bf Update done here.}\\
}

\newcommand{\JPJ}[1]{}

\newcommand{\Polytype}{{\cal T}_{\Sort}}

\newcommand{\gtchorpoX}[1]{\mathop{\gtpo}^{#1}}
\newcommand{\gechorpoX}[1]{\mathop{\gepo}^{#1}}
\newcommand{\gtchorpoXtype}[1]{\mathop{\gtpo}^{#1}_{\Type}}
\newcommand{\gechorpoXtype}[1]{\mathop{\gepo}^{#1}_{\Type}}

\newcommand{\gtsubtarrow}{{~\gtsubt}_{\ra}~}

\newcommand{\newgtS}{\gtS^{\ra}}

\newcommand{\gtacc}{\gtsubt\!\!_{acc}}

\newcommand{\gtchorpoXacc}[1]{\gtchorpoX{#1}_{acc}}

\comment{

\newcommand{\Update}[1]{
{\bf Update done here.}\\
}

\newcommand{\JPJ}[1]{}

\newcommand{\Polytype}{{\cal T}_{\Sort}}

\newcommand{\gtchorpoX}[1]{\mathop{\gtpo}^{#1}}
\newcommand{\gechorpoX}[1]{\mathop{\gepo}^{#1}}
\newcommand{\gtchorpoXtype}[1]{\mathop{\gtpo}^{#1}_{\Type}}
\newcommand{\gechorpoXtype}[1]{\mathop{\gepo}^{#1}_{\Type}}

\newcommand{\gtsubtarrow}{{~\gtsubt}_{\ra}~}

\newcommand{\newgtS}{\gtS^{\ra}}

\newcommand{\gtacc}{\gtsubt\!\!_{acc}}

\newcommand{\gtchorpoXacc}[1]{\gtchorpoX{#1}_{acc}}

}

\begin{document}

\title{\Large\bf The Computability Path Ordering:\\
the End of a Quest}

\author{
  Fr\'ed\'eric Blanqui\inst{1}
  \and
  Jean-Pierre Jouannaud\inst{2}
  \and
  Albert Rubio\inst{3}}
 
\institute{
  {INRIA, Campus Scientifique, BP 239, 54506 Vand{\oe}uvre-l\`es-Nancy Cedex, France}
\and
  {LIX, Projet INRIA TypiCal, \'{E}cole Polytechnique and CNRS,
  91400 Palaiseau, France}
\and
  {Technical University of Catalonia,
  Pau Gargallo 5, 08028 Barcelona, Spain}  
}

\comment{
This work was partly supported by the RNTL
project AVERROES, France-Telecom, and the CICYT project LOGICTOOLS,
ref.\ TIN2004-07925.
}

\pagestyle{plain}
\bibliographystyle{plain}
\maketitle

\thispagestyle{empty}
\large

\begin{abstract}
  In this paper, we first briefly survey automated termination proof
  methods for higher-order calculi. We then concentrate on the
  higher-order recursive path ordering, for which we provide an
  improved definition, the  Computability Path Ordering.
  This new definition appears indeed to capture the essence of
  computability arguments \emph{\`a la Tait and Girard}, therefore
  explaining the name of the improved ordering.

\end{abstract}

%

\section{Introduction}
\label{s:introduction}

This paper addresses the problem of automating termination proofs for
typed higher-order calculi.

The first attempt we know of goes back to Breazu-Tannen and Gallier
\cite{breazu89icalp} and Okada \cite{okada89issac}.  Following up a
pioneering work of Breazu\-Tannen who considered the confluence of
such calculi \cite{breazu88lics}, both groups of authors showed
independently that proving strong normalization of a polymorphic
lambda-calculus with first-order constants defined by first-order
rewrite rules was reducible to the termination proof of the set of
rewrite rules: beta-reduction need not be considered. Both works used
Girard's method based on {\em reducibility candidates} -also called
sometimes {\em computability predicates}. They then gave rise to a
whole new area, by extending the type discipline, and by extending the
kind of rules that could be taken care of.

The type discipline was extended soon later independently by Barbanera
and Dougerthy in order to cover the whole calculus of constructions
\cite{barbanera90ctrs,dougherty92ic}.

Higher-order rewrite rules satisfying the {\em general schema}, a
generalization of G\"odel's primitive recursion rules for higher
types, were then introduced by Jouannaud and Okada
\cite{jouannaud91lics,jouannaud97tcs} in the case of a polymorphic
type discipline. The latter work was then extended first by Barbanera
and Fernandez \cite{barbanera93tlca,barbanera93icalp} and finally by
Barbanera, Fernandez and Geuvers to cover the whole calculus of
constructions \cite{barbanera94lics}.

It turned out that recursors for {\em simple} inductive types could be
taken care of by the general schema, but arbitrary strict inductive
types could not, prompting for an extension of the schema, which was
reformulated for that purpose by Blanqui, Jouannaud and Okada
\cite{blanqui02tcs}. This new formulation was based on the notion of
{\em computability closure} of a term $f(\vect{s})$ headed by a
higher-order constant $f$, defined as a set containing the immediate
subterms $\vect{s}$ of $f(\vect{s})$ and closed under computability
preserving operations in the sense of Tait and Girard. Membership to
the general schema was then defined for an arbitrary rewrite rule as
membership of its right-hand side to the computability closure of its
left-hand side.

Besides being elegant, this formulation was indeed much more flexible
and powerful. By allowing for more expressive rules {\em at the object
level} of the calculus of constructions, it could handle many more
inductive types than originally. The general schema was finally
extended by Blanqui in a series of papers by allowing for {\em
recursive rules on types}, in order to cover the entire calculus of
inductive constructions including strong elimination rules
\cite{blanqui05mscs,blanqui05fi}.

The definition of the general schema used a precedence on higher-order
constants, as does Dershowitz recursive path ordering for first-order
terms \cite{dershowitz82tcs}. This suggested generalizing this
ordering to the higher-order case, a work done by Jouannaud and Rubio
in the case of a simple type discipline under the name of HORPO
\cite{jouannaud99lics}. Comparing two terms with HORPO starts by
comparing their types under a given well-founded quasi-ordering on
types before to proceed recursively on the structure of the compared
terms, comparing first in the precedence the higher-order constants
heading both terms.  Following the recursive path ordering tradition,
a subterm of the left-hand side could also be compared with the whole
right-hand side, regardless of the precedence on their heads.

HORPO was then extended to cover the case of the calculus of
constructions by Walukiewicz \cite{walukiewicz03jfp}, and to use
semantic interpretations of terms instead of a precedence on function
symbols by Borralleras and Rubio \cite{borralleras01lpar}.  HORPO was
also improved by the two original authors in two different ways: by
comparing in the so-called subterm case an arbitrary term belonging to
the computability closure of the left-hand side term with the right-hand
side term, therefore generalizing both HORPO and the general schema;
and by allowing for a restricted polymorphic discipline
\cite{jouannaud07jacm}.  An axiomatic presentation of the rules
underlying HORPO can be found in \cite{goubault01csl}. A more recent
work in the same direction is \cite{dershowitz08pc}.

The ordering and the computability closure definitions turn out to
share many similar constructs, raising expectations for a simpler and
yet more expressive definition, instead of a pair of mutually
inductive definitions for the computability closure and the ordering
itself, as advocated in \cite{blanqui06lpar-horpo}. These expectations
were partly met, on the one hand in \cite{blanqui07jpbirthday} with a
single computability oriented definition, and on the other hand in
\cite{blanqui07lpar} where a new, syntax oriented recursive definition
was given for HORPO. In contrast with the previous definitions, bound
variables were handled explicitly by the ordering, allowing for
arbitrary abstractions in the right-hand sides.

A third, different line of work was started by van de Pol and
Schwichtenberg, who aimed at (semi)-automating termination proofs of
higher-order rewrite rules based on higher-order pattern matching, a
problem generally considered as harder as the previous one
\cite{vandepol93hoa,vandepol95tlca,vandepol96phd}. Related attempts
with more automation appear in \cite{loria92ctrs,jouannaud98tcs}, but
were rather unconclusive for practical applications. The general
schema was then adapted by Blanqui to cover the case of higher-order
pattern matching \cite{blanqui00rta}. Finally, Jouannaud and Rubio
showed how to turn any well-founded ordering on higher-order terms
including beta and eta, into a well-founded ordering for proving
termination of such higher-order rules, and introduced a very simple
modification of HORPO as an application of this result
\cite{jouannaud06rta-termin}.

A fourth line of work was started by Borralleras and Rubio. Among
other material, Borralleras thesis \cite{borralleras03phd} contained a
constraint-based approach to the semantic path ordering \cite{kamin80}
which was shown to encompass the dependency pairs method of Arts and
Giesl \cite{arts00tcs,giesl04lpar} in all its various aspects. Besides
the thesis itself, the principles underlying this work are also
described in \cite{borralleras01lpar} and
\cite{borralleras07jpbirthday}. An interesting aspect is that they
lift to the higher-order case. Extending the dependency pairs method
to the higher-order case was also considered independently by Sakai {\em et
al} \cite{sakai01tis,sakai05tis} and Blanqui \cite{blanqui06wst-hodp}.

Finally, a last line of work addresses the question of proving
termination of higher-order programs. This is of course a slightly
different question, usually addressed by using abstract
interpretations. These interpretations may indeed use the general
schema or HORPO as a basic ingredient for comparing inputs of a
recursive call to those of the call they originate from. This line of
work includes
\cite{hughes96popl,chin01hosc,benamram01popl,xi02hosc,abel04ita,barthe04mscs,blanqui04rta,giesl06rta}. An 
important related work, considering pure lambda terms,
is~\cite{jones04rta}.

We believe that our quest shall be shown useful for all these lines of
work, either as a building block, or as a guiding principle.

In this paper, we first slightly improve the definition of HORPO in
the very basic case of a simple type discipline, and rename it as the
Computability Path Ordering.  We then address the treatment of
inductive types which remained {\em ad hoc} so far, therefore concluding our
quest thanks to the use of accessibility, a relationship which was
shown to generalize the notion of inductive type by Blanqui
\cite{blanqui05mscs,blanqui05fi}. We finally list which are the most important
question to be addressed for those who would like to start a new quest.

\section{Higher-Order Algebras}
\label{s:preliminaries}

Polymorphic higher-order algebras are introduced
in~\cite{jouannaud07jacm}. Their purpose is twofold: to define a
simple framework in which many-sorted algebra and typed
lambda-calculus coexist; to allow for polymorphic types for both
algebraic constants and lambda-calculus expressions. For the sake of
simplicity, we will restrict ourselves to monomorphic types in this
presentation, but allow us for polymorphic examples. Carrying out the
polymorphic case is no more difficult, but surely  more painful.

We give here the minimal set of notions to be reasonably
self-contained.

Given a set $\Sort$ of {\em sort symbols} of a fixed arity, denoted by
$s:*^n\ra *$, the set of {\em types} is generated by the constructor
$\ra$ for {\em functional types}:

\[
\begin{array}{c}
\Polytype := s(\Polytype^n) ~|~ (\Polytype \ra \Polytype)\\
\mbox{for $s: *^n\ra *~\in\Sort$}
\end{array}
\] 

Function symbols are meant to be algebraic operators equiped with a
fixed number $n$ of arguments (called the {\em arity}) of respective
types $\sigma_1,\ldots,\sigma_n$, and an
\emph{output type} $\sigma$. Let ${\F} = \biguplus_{\sigma_1,\ldots,\sigma_n,\sigma}
\F_{\sigma_1 \times\ldots \times \sigma_n \ra \sigma}$.
The membership of a given function symbol $f$ to $\F_{\sigma_1
\times\ldots \times \sigma_n \ra \sigma}$ is called a {\em type
declaration} and written $f : \sigma_1 \times\ldots \times \sigma_n
\ra \sigma$.

The set $\TFX$ of {\em raw algebraic $\lambda$-terms} is generated from
the signature ${\cal F}$ and a denumerable set $\X$ of
variables according to the grammar:
\[\Term := 
   \X ~|~ (\lambda \X : \Polytype . \Term) ~|~ @(\Term,\Term) ~|~ \F(\Term,\ldots,\Term).
\] 
The raw term $\lambdax:\sigma.u$ is an {\em abstraction} and $@(u,v)$
is an application.  We may omit $\sigma$ in $\lambdax:\sigma.u$ and
write $@(u,v_1,\ldots,v_n)$ or $u(v_1,\ldots,v_n)$, $n>0$, omitting
applications.  $\Var{t}$ is the set of free variables of $t$. A raw
term $t$ is \emph{ground} if $\Var{t}=\emptyset$.  The notation
$\vect{s}$ shall be ambiguously used for a list, a multiset, or a set
of raw terms $s_1,\ldots,s_n$.

Raw terms are identified with finite labeled trees by considering
$\lambdax :\sigma.u$, for each variable $x$ and type $\sigma$, as a
unary function symbol taking $u$ as argument to construct the raw term
$\lambdax : \sigma . u$.  {\em Positions} are strings of positive
integers.  $t|_p$ denotes the {\it subterm} of $t$ at position
$p$. We use $t\gesubt t|_p$ for the subterm relationship. The result
of replacing $t|_p$ at position $p$ in $t$ by $u$ is written $t[u]_p$.

Typable raw terms are called \emph{terms}. The typing judgements are
standard.  We categorize terms into three disjoint classes:
\begin{enumerate}
\item
\emph{Abstractions} headed by $\lambda$;
\item
\emph{Prealgebraic} terms headed by a function symbol, assuming (for
the moment) that the output type of $f\in\F$ is a base type;
\item
\emph{Neutral} terms are variables or headed by an application.
\end{enumerate}

Substitutions, rewrite rules and higher-order reduction orderings are
as expected, see~\cite{jouannaud07jacm}.

\section{The Computability Path Ordering}
\label{s:ordering}

CPO is generated from three basic ingredients: a {\em type ordering};
a {\em precedence} on functions symbols; and a {\em status} for the
function symbols. Accessibility is an additionnal ingredient originating in
inductive types, while the other three were already needed for
defining HORPO.  We describe these ingredients before defining the
computability path ordering. We define the ordering in two steps,
accessibility being used in the second step only. The first ordering
is therefore simpler, while the second is more expressive.


\subsection{Basic ingredients}
\label{ssto}

\begin{itemize}
\item
a {\em precedence} $\geF$ on symbols in $\F\cup\{@\}$, with $f\gtF @$
for all $f\in\F$.

\item
a status for symbols in $\F\cup\{@\}$ with $@\in Mul$.

\item
and a quasi-ordering on types $\geS$ called {\em the type ordering}
satisfying the following properties, where $\eqS$ denotes its
associated equivalence relation $\geS\cap\leS$ and $\gtS$ its strict
part $\geS\setminus\leS$:

\begin{enumerate}
\item
\label{wf}
{\em Well-foundedness}:
${\newgtS}={\gtS\cup\gtsubtarrow}$ is well-founded,\\ where
$\sigma\ra\tau\gtsubtarrow\sigma$;
\item
\label{ma}
{\em Right arrow subterm}:
$\sigma\ra\tau\gttype\tau$;
\item
\label{poa}
{\em Arrow preservation}: $\tau\ra\sigma\eqS\alpha$ iff
$\alpha=\tau'\ra\sigma'$, $\tau'\eqS\tau$ and $\sigma\eqS\sigma'$;
\item
\label{sa}
{\em Arrow decreasingness}: $\tau\ra\sigma\gtS\alpha$ implies
$\sigma\geS\alpha$ or else $\alpha=\tau'\ra\sigma'$, $\tau'\eqS\tau$
and $\sigma\gtS\sigma'$;
\end{enumerate}

\end{itemize}


Arrow preservation and decreasingness imply that the type ordering
does {\em not}, in general, have the left arrow subterm property:
$\sigma\ra\tau\not\!\getype \sigma$. A first axiomatic definition of
the type ordering was given in \cite{jouannaud01draft}, which did not
need right arrow subterm. A new one, expected to be easier to
understand, was given in \cite{jouannaud07jacm} based solely on
$\getype$, which uses another axiom, {\em arrow monotonicity}, to
force the right arrow subterm property. As pointed out to us recently,
this set of axioms is unfortunately inconsistent
\cite{femke08pc}. However, the restriction of the recursive path
ordering proposed there for a type ordering does not satisfy arrow
monotonicity, but does satisfy instead the corrected set of axioms
given here.

We now give two important properties of the type ordering:

\begin{lemma}\cite{jouannaud07jacm}
\label{l:pdt}
Assuming $\sigma\eqS \tau$, $\sigma$ is a data type iff
$\tau$ is a data type.
\end{lemma}

\begin{lemma}
\label{l:geS}
If $\alpha\ra\sigma\geS\beta\ra\tau$ then $\sigma\geS\tau$.
\end{lemma}

\begin{proof}
If $\alpha\ra\sigma\eqS\beta\ra\tau$ then, by arrow preservation,
$\alpha\eqS\beta$ and $\sigma\eqS\tau$. If
$\alpha\ra\sigma\gtS\beta\ra\tau$, then, by arrow decreasingness,
either $\alpha\eqS\beta$ and $\sigma\gtS\tau$, or else
$\sigma\gtS\beta\ra\tau$. In the latter case, $\beta\ra\tau\gtS\tau$ by
right arrow subterm and we conclude by transitivity.\cqfd
\end{proof}


\subsection{Notations}

Our ordering notations are as follows:

\begin{itemize}
\item
$s\gtchorpoX{X} t$ for the main ordering, with a finite set of
variables $X\subset\X$ and the convention that $X$ is omitted when
empty;
\item
$s:\sigma\gtchorpoXtype{X} t:\tau$ for $s\gtchorpoX{X} t$ and
$\sigma\getype\tau$;
\item
$l:\sigma\gtchorpoXtype{}r:\tau$ as initial call for each $l\ra r\in
R$;
\item
$s\succ\vect{t}$ is a shorthand for $s\succ u$ for all $u\in\vect{t}$;
\item
$\succeq$ is the reflexive closure of $\succ$.
\end{itemize}

We can now introduce the definition of CPO.


\subsection{Ordering definition}

\begin{definition}
\label{d:ordering-one}
$s:\sigma\gtchorpoX{X}t:\tau$ iff either:

\begin{enumerate}

\item
\label{funleft1}
$s=f(\vect{s})$ with $f\in\F$ and either of
\begin{enumerate}
\item
\label{funvar1}
$t\in X$
\item
\label{funstat1}
$t=g(\vect{t})$ with $f=_\F g\in\F$, $s\gtchorpoX{X}\vect{t}$ and
$\vect{s} (\gtchorpoXtype{})_{stat_f} \vect{t}$
\item
\label{funprec1}
$t=g(\vect{t})$ with $f\gtF g\in\F\cup\{@\}$ and
$s\gtchorpoX{X}\vect{t}$
\item
\label{funabs1}
$t=\lambday:\beta.w$ and $s\gtchorpoX{X\cup\{z\}}w\{y\mapsto z\}$ for
$z:\beta$ fresh
\item
\label{funsubt1}
$u \gechorpoXtype{}t$ for some $u\in\vect{s}$
\end{enumerate}

\item
\label{appleft1}
$s=@(u,v)$ and either of
\begin{enumerate}
\item
\label{appvar1}
$t\in X$
\item
\label{appstat1}
$t=@(u',v')$ and $\{u,v\}(\gtchorpoXtype{})_{mul} \{u',v'\}$
\item
\label{appabs1}
$t=\lambday:\beta.w$ and $s\gtchorpoX{X}w\{y\mapsto z\}$ for $z:\beta$
fresh
\item
\label{appsubt1}
$u\gechorpoXtype{X}t$ or $v\gechorpoXtype{X}t$
\item
\label{beta1}
$u=\lambdax:\alpha.w$ and $w\{x\mapsto v\}\gechorpoX{X} t$
\end{enumerate}

\item
\label{absleft1}
$s=\lambda x:\alpha.u$ and either of
\begin{enumerate}
\item
\label{absvar1}
$t\in X$
\item
\label{absstat1}
$t=\lambda y:\beta.w$, $\alpha\eqtype\beta$ and
$u\{x\!\mapsto\!z\}\gtchorpoX{X}w\{y\!\mapsto\!z\}$ for $z\!:\!\beta$
fresh
\item
\label{absprec1}
$t=\lambda y:\beta.w$, $\alpha\neqtype\beta$ and $s\gtchorpoX{X}
w\{y\mapsto z\}$ for $z:\beta$ fresh
\item
\label{abssubt1}
$u\{x\mapsto z\}\gechorpoXtype{X} t$ for
$z:\alpha$ fresh

\item
\label{eta1}
$u=@(v,x)$, $x\not\in\Var{v}$ and $v\gechorpoX{X}t$
\end{enumerate}
\end{enumerate}

\end{definition}


Because function symbols, applications and abstractions do not behave
exactly the same, we chosed to organize the definition according to
the left-hand side head symbol: a function symbol, an application, or
an abstraction successively. In all three cases, we first take care of
the case where the right-hand side is a bound variable -case named
{\mbox{\em variable}-,} then headed by a symbol which is the same as (or
equivalent to) the left-hand side head symbol -case {\em status}-, or
headed by a symbol which is strictly smaller in the precedence than
the left-hand side head symbol -case {\em precedence}-, before to go
with the -case {\em subterm}. The precedence case breaks into two
sub-cases when the left-hand side is a function symbol, because
abstractions, which can be seen as smaller than other symbols, need
renaming of their bound variable when pulled out, which makes their
treatment a little bit different formally from the standard precedence
case. There are two specific cases for \emph{application} and
\emph{abstraction}: one for beta-reduction, and one for eta-reduction,
which are both built in the definition.

This new definition schema appeared first in \cite{blanqui07lpar} in a
slightly different format. It incorporates two major innovations with
respect to the version of HORPO defined in \cite{jouannaud07jacm}.
The first is that terms can be ordered without requiring that their
types are ordered accordingly.  This will be the case whenever we can
conclude that some recursive call is terminating by using
computability arguments rather than an induction on types. Doing so,
the ordering inherits directly much of the expressivity of the
computability closure schema used in \cite{jouannaud07jacm}. The
second is the annotation of the ordering by the set of variables $X$
that were originally bound in the right-hand side term, but have
become free when taking some subterm. This allows
rules~\ref{funabs1},~\ref{appabs1} and~\ref{absprec1} to pull out
abstractions from the right-hand side regardless of the left-hand side
term, meaning that abstractions are smallest in the precedence. Among
the innovations with respect to \cite{blanqui07lpar} are
rules~\ref{absprec1}, which compares abstractions whose bound
variables have non-equivalent types, and rule~\ref{appsubt1}, whose
formulation is now stronger.

This definition suffers some subtle limitations:

\begin{enumerate}
\item
Case~\ref{funabs1} uses recursively the comparison
$s\gtchorpoX{X\cup\{z\}}w\{y\mapsto z\}$ for $z$ fresh, implying that
the occurrences of $z$ in $w$ can be later taken care of by
Case~\ref{funvar},~\ref{appvar}~or~\ref{absvar}. This is no
limitation.

Cases \ref{appabs1} and \ref{absprec1} use instead the recursive
comparison \mbox{$s\gtchorpoX{X}w\{y\!\mapsto\!z\}$}, with $z$ fresh, hence
$z\notin X$. As a consequence, these recursive calls cannot succeed if
$z\in\Var{w}$. We could have added this redundant condition for sake
of clarity. We prefered to privilege uniformity and locality of tests.

As a consequence, Cases~\ref{funabs1},~\ref{appabs1} and
\ref{absprec1} cannot be packed together as it
was unfortunately done in~\cite{blanqui07lpar}, where correct proofs
were however given which did of course not correspond to the given
definition.
\item
The subterm case~\ref{funsubt} uses recursively the comparison
$u\gtchorpoXtype{}t$ instead of the expected comparison
$u\gtchorpoXtype{X}t$.

On the other hand, the other subterm definitions, Cases~\ref{appsubt}
and~\ref{abssubt} use the expected comparisons $u\gechorpoXtype{X}t$
or $v\gechorpoXtype{X}t$ in the first case, and $u\{x\mapsto
z\}\gechorpoXtype{X}t$ in the second. This implies again that the
various subterm cases cannot be packed together.
\item
Case~\ref{funstat1} uses recursively the comparison $\vect{s}
({\gtchorpoXtype{}})_{stat_f} \vect{t}$ instead of the stronger
comparison $\vect{s} ({\gtchorpoXtype{X}})_{stat_f} \vect{t}$.
\end{enumerate}

All our restrictions are justified by their use in the
well-foundedness proof of $\gtchorpoXtype{}$. There is an even better
argument: the ordering would not be well-founded otherwise, as can be
shown by means of counter-examples. We give two below.


We start with an example of non-termination obtained when replacing
the recursive call $\vect{s}(\gtchorpoXtype{})_{stat_f}\vect{t}$ by
$\vect{s}(\gtchorpoXtype{X})_{stat_f}\vect{t}$ in Case~\ref{funstat}.

\begin{example}
\label{ex:non-termination1}
Let $a$ be a type, and $\{f:a\times a \ra a, g:(a\ra a)\ra a\}$ be the
signature. Let us consider the following non-terminating rule (its
right-hand side beta-reduces to its left-hand side in one beta-step):

\[f(g(\lambdax.f(x,x)),g(\lambdax.f(x,x)))
\ra @(\lambdax.f(x,x), g(\lambdax.f(x,x)))\]

Let us assume that $f\gtF g$ and that $f$ has a multiset status.  We
now show that the ordering modified as suggested above succeeds with
the goal

\begin{enumerate}
\item
$f(g(\lambdax.f(x,x)),g(\lambdax.f(x,x))) \gtchorpoXtype{}
@(\lambdax.f(x,x), g(\lambdax.f(x,x)))$.
\end{enumerate}

Since type checks are trivial, we will omit them, although the reader
will note that there are very few of them indeed. Our goal yields two
sub-goals by Case~\ref{funprec1}:

\begin{enumerate}
\setcounter{enumi}{1}
\item
\label{one}
$f(g(\lambdax.f(x,x)),g(\lambdax.f(x,x))) \gtchorpoX{} \lambdax.f(x,x)$ and
\item
\label{two}
$f(g(\lambdax.f(x,x)),g(\lambdax.f(x,x))) \gtchorpoX{} g(\lambdax.f(x,x))$.
\end{enumerate}

Sub-goal~\ref{one} yields by Case~\ref{funabs1}

\begin{enumerate}
\setcounter{enumi}{3}
\item
$f(g(\lambdax.f(x,x)),g(\lambdax.f(x,x))) \gtchorpoX{\{z\}}
f(z,z)$ which yields by Case~\ref{funstat1}
\item
$f(g(\lambdax.f(x,x)),g(\lambdax.f(x,x))) \gtchorpoX{\{z\}} z$ twice, solved by Case~\ref{funvar1} and
\item
\label{mulsubgoal}
$\{g(\lambdax.f(x,x)),g(\lambdax.f(x,x))\} (\gechorpoXtype{\{z\}})
\{z,z\}$ solved by Case~\ref{funvar1} applied twice.
\end{enumerate}

We are left with sub-goal~\ref{two} which yields by Case~\ref{funprec1}

\begin{enumerate}
\setcounter{enumi}{6}
\item
$f(g(\lambdax.f(x,x)),g(\lambdax.f(x,x))) \gtchorpoX{}
\lambdax.f(x,x)$, which happens to be the already solved
sub-goal~\ref{one}, and we are done.
\end{enumerate}

With the definition we gave, sub-goal~\ref{mulsubgoal} becomes:\\
$\{g(\lambdax.f(x,x)),g(\lambdax.f(x,x))\} (\gtchorpoXtype{})_{mul}
\{z,z\}$ and does not succeed since the set of previously bound
variables has been made empty.

The reader can check that chosing the precedence $g\gtF f$ yields
exactly the same result in both cases.
\cqfd
\end{example}

Next is an example of non-termination due to Cynthia Kop and
Femke van Raamsdong~\cite{femke08pc}, obtained when replacing the
recursive call $s\gtchorpoX{X}w\{y\mapsto z\}$ by
$s\gtchorpoX{X\cup\{z\}}w\{y\mapsto z\}$ in Case~\ref{appabs}.

\begin{example} 
\label{ex:nontermination2}
Let $o$ be a type, and $\{f:o\ra o, A:o, B: o\ra o\ra o\}$ be the
signature. Let us consider the following non-terminating set of rules:

\[\begin{array}{rclr}
@(@(B, A),A) &\ra& @(\lambda z:o.f(z), A) &\quad\quad (1)\\
f(A) &\ra& @(@(B,A),A)&(2)
\end{array}
\]
since
\[@(@(B, A),A) \lrps{}{1} @(\lambda z:o.f(z), A)\lrps{}{\beta}f(A)\lrps{}{2}@(@(B,A),A)\]

Let us assume that $A\gtF f\gtF B$ and consider the goals:

\begin{enumerate}
\item
\label{goal1}
$@(@(B,A),A):o \gtchorpoXtype{} @(\lambda z:o.f(z), A):o$, and
\item
\label{goal2}
$f(A):o \gtchorpoXtype{} @(@(B,A),A):o$.
\end{enumerate}

Goal~\ref{goal1} yields two sub-goals by Case~\ref{appstat}:

\begin{enumerate}
\setcounter{enumi}{2}
\item
\label{goal12}
$A:o \gechorpoXtype{} A:o$, which succeeds trivially, and 
\item
\label{goal11}
$@(B,A): o\ra o\gtchorpoXtype{} \lambda z:o.f(z) : o\ra o$ which yields by
modified Case~\ref{appabs1}:
\item
\label{goal111}
$@(B,A) \gtchorpoX{\{z\}} f(z)$, which yields in turn by Case~\ref{appsubt}
\item
\label{goal1111}
$A:o  \gtchorpoXtype{\{z\}} f(z):o$ which yields by Case~\ref{funprec}
\item
$A:o  \gtchorpoXtype{\{z\}} z:o$ which succeeds by Case~\ref{funvar}.
\end{enumerate}

Note that we have used $B$ for its large type, and $A$ for eliminating
$f(z)$, exploiting a kind of divide and conquer ability of the ordering.
We are left with goal~\ref{goal2} which yields two subgoals by
Case~\ref{funabs1}

\begin{enumerate}
\setcounter{enumi}{7}
\item
$f(A)\gtchorpoX{} A$ which succeeds by Case~\ref{funsubt}, and
\item
$f(A)\gtchorpoX{} @(B,A)$, which yields by Case~\ref{funprec}:
\item
$f(A)\gtchorpoX{} A$, which succeeds by Case~\ref{funsubt}, and
\item
$f(A)\gtchorpoX{} B$, which succeeds by Case~\ref{funprec}, therefore
ending the computation.  \cqfd
\end{enumerate}
\end{example}

More examples justifying our claim that the quest has come to en end
are given in the full version of this paper.
\comment{
We could continue the list of example showing that other similar
potential improvements of the ordering yield a non-terminating
relation. Forf example, in Case~\ref{funsubt}, it is both necessary to
check types in the recursive call, and to have an empty set $X$ of
previously bound variables. Similarly, the type-cheks in
Cases~\ref{appsubt} and ~\ref{abssubt} are necessary.
}

We give now an example of use of the computability path ordering with the
inductive type of Brouwer's ordinals, whose constructor $lim$ takes an
infinite sequence of ordinals to build a new, limit ordinal, hence
admits a functional argument of type $\Nat\ra O$, in which $O$ occurs
positively.  As a consequence, the recursor admits a more complex
structure than that of natural numbers, with an explicit abstraction
in the right-hand side of the rule for $lim$. The strong normalization
proof of such recursors is known to be hard.


\begin{example}
Brouwer's ordinals.

\noindent
$0:O \quad\quad\quad S:O\ra O \quad\quad\quad lim:(\Nat\ra O)\ra O$\\
$rec : O\times\alpha\times(O\ra\alpha\ra\alpha)\times
((\Nat\ra O)\ra(\Nat\ra\alpha)\ra\alpha)\ra\alpha$

The rules defining the recursor on Brouwer's ordinals are:

$rec(0,U,X,W)\ra U$\\
$rec(S(n),U,X,W)\ra @(X,n,rec(n,U,X,W))$\\
$rec(lim(F),U,X,W)\ra @(W,F,\lambda n.rec(@(F,n),U,X,W))$

Let us try to prove that the third rule is in $\gtchorpoXtype{}$.

\begin{enumerate}
\item
$s=rec(lim(F),U,X,W) \gtchorpoXtype{} @(W,F,\lambda
n.rec(@(F,n),U,X,W))$ yields 4 sub-goals according to
Case~\ref{funprec1}:
\item
$\alpha \getype \alpha$ which is trivially satisfied, and
\item
$s\gtchorpoX{} \{W, F, \lambda n.rec(@(F,n), U, X, W)\}$
which simplifies to:
\item
$s\gtchorpoX{} W$ which succeeds by Case~\ref{funsubt1},
\item
$s\gtchorpoX{} F$, which generates by Case~\ref{funsubt1} the
comparison $lim(F)\gtchorpoXtype{} F$ which fails since $lim(F)$ has a
type which is strictly smaller than the type of $F$.
\item
$s\gtchorpoX{} \lambda n.rec(@(F,n), U, X, W)$ which yields by 
Case~\ref{funabs1}
\item
$s\gtchorpoX{\{n\}} rec(@(F,n), U, X, W)$ which yields by 
Case~\ref{funstat1}
\item
$\{lim(F),U,X,W\} (\gtchorpoXtype{})_{mul} \{@(F,n),U,X,W\}$, which
reduces to
\item
$lim(F) \gtchorpoXtype{} @(F,n)$, whose type comparison succeeds,
yielding by Case~\ref{funprec1}
\item
$lim(F) \gtchorpoX{} F$ which succeeds by Case~\ref{funsubt1}, and
\item
$lim(F) \gtchorpoX{} n$ which fails because track of $n$ has been
lost!
\end{enumerate}
\end{example}

Solving this example requires therefore: first, to access directly the
subterm $F$ of $s$ in order to avoid the type comparison for $lim(F)$
and $F$ when checking recursively whether the comparison
$s\gtchorpoX{}\lambda n.rec(@(F,n)$, $U, X, W)$ holds; and second, to
keep track of $n$ when comparing $lim(F)$ and $n$.


\subsection{Accessibility}
\label{ss:accessibility}

While keeping the same type structure, we make use here of a fourth
ingredient, the {\em accessibility} relationship for data types
introduced in \cite{blanqui00rta}. This will allow us to solve
Brouwer's example, as well as other examples of non-simple inductive
types.

We say that a data type is \emph{simple} is it is a type constant.  We
restrict here our definition of accessibility to simple data types.
To this end, we assume that all type constructors are constants, that
is, have arity zero. We can actually do a little bit more, assuming
that simple data types are not greater or equal (in $\getype$) to
non-constant data types, allowing the simple data types to live in a
separate world.

The sets of {\em positive and negative positions} in a type $\sigma$
are inductively defined as follows:

\begin{itemize}
\item
$Pos^+(\sigma)=\{\varepsilon\}$ if $\sigma$ is a simple data type
\item
$Pos^-(\sigma)=\emptyset$ if $\sigma$ is a simple data type
\item
$Pos^\delta(\sigma\ra\tau)=1\cdot Pos^{-\delta}(\sigma)\cup 2\cdot
Pos^\delta(\tau)$\\
where $\delta\in\{+,-\}$, ${-+}={-}$ and ${--}=+$ (usual rules
of signs)
\end{itemize}

Then we say that a simple data type $\sigma$ occurs (only) {\em positively}
in a type $\tau$ if it occurs only at positive positions:
$Pos(\sigma,\tau)\subseteq Pos^+(\tau)$, where $Pos(\sigma,\tau)$ is
the set of positions of the occurrences of $\sigma$ in $\tau$.

The set $Acc(f)$ of {\em accessible argument positions} of a function
symbol $f:\sigma_1\ldots\sigma_n\ra\sigma$, where $\sigma$ is a simple
data type, is the set of integers $i\in\{1,\ldots,n\}$ such that:

\begin{itemize}
\item
no simple data type greater than $\sigma$ occurs in $\sigma_i$,
\item
simple data types equivalent to $\sigma$ occurs only positively in
$\sigma_i$.
\end{itemize}

Then a term $u$ is {\em accessible} in a term $v$, written $v\gtacc
u$, iff $v$ is a pre-algebraic term $f(\vect{s})$ and there exists
$i\in Acc(f)$ such that either $u=s_i$ or $u$ is accessible in $s_i$
($\gtacc$ is transitive).

A term $u$ is accessible in a sequence of terms $\vect{v}$ iff it is
accessible in some $v\in\vect{v}$, in which case we write
$\vect{s}\gtacc u$. Note that the terms accessible in a term $v$ are
strict subterms of $v$.


We can now obtain a more elaborated ordering as follows:

\begin{definition}
\label{d:ordering-two}
$s:\sigma\gtchorpoX{X}t:\tau$ iff either:

\begin{enumerate}

\item
\label{funleft}
$s=f(\vect{s})$ with $f\in\F$ and either of
\begin{enumerate}
\item
\label{funvar}
$t\in X$
\item
\label{funstat}
$t=g(\vect{t})$ with $f=_\F g\in\F$, $s\gtchorpoX{X}\vect{t}$ and
$\vect{s}(\gtchorpoXtype{}\cup \gtchorpoXacc{X,s})_{stat_f}\vect{t}$
\item
\label{funprec}
$t=g(\vect{t})$ with $f\gtF g\in\F\cup\{@\}$ and
$s\gtchorpoX{X}\vect{t}$
\item
\label{funabs}
$t=\lambday:\beta.w$ and $s\gtchorpoX{X\cup\{z\}}w\{y\mapsto z\}$ for
$z:\beta$ fresh
\item
\label{funsubt}
$u \gechorpoXtype{}t$ for some $u\in\vect{s}$
\item
\label{funacc}
$u \gechorpoXtype{}t$ for some $u$ such that $\vect{s}\gtacc u$
\end{enumerate}

\item
\label{appleft}
$s=@(u,v)$ and either of
\begin{enumerate}
\item
\label{appvar}
$t\in X$
\item
\label{appstat}
$t=@(u',v')$ and $\{u,v\}(\gtchorpoXtype{})_{mul} \{u',v'\}$
\item
\label{appabs}
$t=\lambday:\beta.w$ and $s\gtchorpoX{X}w\{y\mapsto z\}$
for $z:\beta$ fresh
\item
\label{appsubt}
$u\gechorpoXtype{X}t$ or $v\gechorpoXtype{X}t$
\item
\label{beta}
$u=\lambdax:\alpha.w$ and $w\{x\mapsto v\}\gechorpoX{X} t$
\end{enumerate}

\item
\label{absleft}
$s=\lambda x:\alpha.u$ and either of
\begin{enumerate}
\item
\label{absvar}
$t\in X$
\item
\label{absstat}
$t=\lambda y:\beta.w$, $\alpha\eqtype\beta$ and
$u\{x\!\mapsto\!z\}\gtchorpoX{X}w\{y\!\mapsto\!z\}$ for $z\!:\!\beta$
fresh
\item
\label{absprec}
$t=\lambda y:\beta.w$, $\alpha\neqtype\beta$ and
$s\gtchorpoX{X}w\{y\mapsto z\}$ for $z:\beta$ fresh
\item
\label{abssubt}
$u\{x\mapsto z\}\gechorpoXtype{X} t$ for
$z:\alpha$ fresh
\item
\label{eta}
$u=@(v,x)$, $x\not\in\Var{v}$ and $v\gechorpoX{X}t$
\end{enumerate}
\end{enumerate}

\noindent
where $u:\sigma\gtchorpoXacc{X,s}t:\tau$ iff $\sigma\geS\tau$,
$t=@(v,\vect{w})$, $u\gtacc v$ and $s\gtchorpoX{X}\vect{w}$.
\end{definition}

The only differences with the previous definition are in
Case~\ref{funstat} of the main definition which uses an additional
ordering $\gtchorpoXacc{X,s}$ based on the accessibility relationship
$\gtacc$ to compare subterms headed by equivalent function symbols,
and in Case~\ref{funacc} which uses the same relationship $\gtacc$ to
reach deep subterms that could not be reached otherwise. Following up
a previous discussion, notice that we have kept the same formulation
in Cases~\ref{appabs} and~\ref{absprec}, rather than use the easier
condition $y\not\in\Var{w}$.

We could of course strengthen $\gtchorpoXacc{X,s}$ by giving
additional cases, for handling abstractions and function symbols on
the right~\cite{blanqui00rta,blanqui07jpbirthday}.  We could also
think of improving Case~\ref{funsubt} by replacing $\vect{s}\gtacc u$
by the stronger condition $\vect{s} \gtchorpoXacc{X,s} u$. We have not
tried these improvements yet.

We now revisit Brouwer's example, whose strong normalization proof is
checked automatically by this new version of the ordering:


\begin{example}
Brouwer's ordinals.

We skip goals 2,3,4 which do not differ from the previous attempt.

\begin{enumerate}
\item
$s=rec(lim(F),U,X,W) \gtchorpoXtype{} @(W, F, \lambda n . rec(@(F,n),
U, X, W))$ yields 4 sub-goals according to Case~\ref{funprec1}:
\setcounter{enumi}{4}
\item
$s\gtchorpoX{} F$, which succeeds now by Case~\ref{funacc},
\item
$s\gtchorpoX{} \lambda n.rec(@(F,n), U, X, W)$ which yields by 
Case~\ref{funabs1}
\item
$s\gtchorpoX{\{n\}} rec(@(F,n), U, X, W)$ which
yields goals 8 and 12 by Case~\ref{funstat1}
\item
$\{lim(F), U, X, W\}
(\gtchorpoXtype{}\cup\gtchorpoXacc{\{n\},s})_{mul} \{@(F,n),U,X,W\}$,
which reduces to
\item
$lim(F) \gtchorpoXacc{\{n\},s} @(F,n)$ which succeeds since $O\eqtype
O$, $F$ is accessible in $lim(F)$ and $s\gtchorpoX{\{n\}}n$ by case
Case~\ref{funvar}. Our remaining goal
\item
$s\gtchorpoX{\{n\}} \{@(F,n),U,X,W\}$\\ decomposes into three goals
trivially solved by Case~\ref{funsubt}, that is
\item
$s\gtchorpoX{\{n\}} \{U,X,W\}$, and one additional goal
\item
$s\gtchorpoX{\{n\}} @(F,n)$ which yields two goals by Case~\ref{funprec}
\item
$s\gtchorpoX{\{n\}} F$, which succeeds by Case~\ref{funacc}, and
\item
$s\gtchorpoX{\{n\}} n$ which succeeds by
Case~\ref{funvar}, thus ending the computation.
\end{enumerate}
\end{example}

\section{Conclusion}
\label{s:conclusion}

An implementation of CPO with examples is available from the web page
of the third author.

There is still a couple of possible improvements that deserve to be
explored thoroughly: change -if possible at all- the recursive calls
of Cases~\ref{funsubt},~\ref{appabs} and~\ref{absprec} of the
definition of CPO as discussed in Section~\ref{s:ordering}; ordering
$\F\cup\{@\}$ arbitrarily -this would be useful for some examples,
e.g., some versions of Jay's pattern calculus \cite{jay04toplas};
increasing the set of accessible terms; and improve the definition of
the accessibility ordering $\gtchorpoXacc{X}$, possibly by making it
recursive.

A more challenging problem to be investigated then is the
generalization of this new definition to the calculus of constructions
along the lines of \cite{walukiewicz03jfp} and the suggestions made
in \cite{jouannaud07jacm}, where an RPO-like ordering on types was
proposed which allowed to give a single definition for terms and
types. Starting this work with definition~\ref{d:ordering-one} is of
course desirable.

Finally, it appears that the recursive path ordering and the
computability closure are kind of dual of each other: the definitions
are quite similar, the closure constructing a set of terms while the
ordering deconstructs terms to be compared, the basic case being the
same: bound variables and various kinds of subterms. Besides, the
properties to be satisfied by the type ordering, which were infered
from the proof of the computability predicates, almost characterize a
recursive path ordering on the first-order type structure.  An
intriguing, challenging question is therefore to understand the
precise relationship between computability predicates and path
orderings.

\medskip

{\bf Acknowledgements:} the second author wishes to point out the
crucial participation of Mitsuhiro Okada to the very beginning of this
quest, and to thank Makoto Tatsuta for inviting him in december 2007
at the National Institute for Informatics in Tokyo, whose support
provided him with the ressources, peace and impetus to conclude this
quest with his coauthors.  We are also in debt with Cynthia Kop and
Femke van Raamsdonk for pointing out to us a (hopefully minor) mistake
in published versions of our work on HORPO.

\bibliography{main}

\end{document}